\documentclass[a4paper,twocolumn,11pt,unpublished]{quantumarticle}
\pdfoutput=1
\usepackage[utf8]{inputenc}
\usepackage[english]{babel}
\usepackage[T1]{fontenc}

\usepackage{amssymb}
\usepackage{amsmath}
\usepackage{amsthm}
\newtheorem{lem}{Lemma}
\newtheorem{thm}{Theorem}

\usepackage{graphicx}
\usepackage{physics}

\usepackage[unicode=true, breaklinks=true, hidelinks]{hyperref}

\usepackage[numbers]{natbib}

\renewcommand{\char}{\ensuremath{\mathrm{char}}}

\begin{document}

\title{A Simplified Expression for Quantum Fidelity}

\author{Adrian M\"uller}
\email{adrian.mueller@aalto.fi}
\affiliation{Department of Computer Science, Aalto University, 02150 Espoo, Finland}
\orcid{0009-0005-7744-5950}

\maketitle

\begin{abstract}
Quantum fidelity is one of the most important measures of similarity between mixed quantum states. However, the usual formulation is cumbersome and hard to understand when encountering the first time. This work shows in a novel, elegant proof that the expression can be rewritten into a form, which is not only more concise but also makes its symmetry property more obvious. Further, the simpler expression gives rise to a formulation that is subsequently shown to be more computationally efficient than the best previous methods by avoiding any full decomposition. Future work might look for ways in which other theorems could be affected or utilize the reformulation where fidelity is the computational bottleneck.
\end{abstract}

\section{Introduction}

One of the most fundamental tasks in quantum information processing is the ability to tell the similarity or closeness of two quantum states. It is important in a wide range of applications, like quantum metrology, quantum machine learning, quantum communication, quantum cryptography, or quantum thermodynamics. For example, a similarity measure would be used to assess how much a message was disturbed when transported over distance \cite{Barrett2004DeterministicQT}, or to characterize phase transitions in quantum systems where the ground state might abruptly change \cite{Gu2008FidelityAT}.

A common method for this purpose is the quantum fidelity, also known as Uhlmann fidelity or Uhlmann-Jozsa fidelity, which is capable to assess the similarity of a pair of mixed states. However, the usual formulation of this most general form of quantum fidelity, where both quantum states are mixed states, has several drawbacks. One of the most important drawbacks is that is it computationally expensive, and might become a bottleneck if it has to be calculated often and on large density matrices \cite{liang2019quantum}. Another drawback is that its symmetry property is not immediately obvious and it might be hard to grasp on the first encounter \cite{nielsen2010quantum}.

Previous work has already shown that there exists a simpler expression that is equivalent to the usual formulation \cite{li2015some, baldwin2023efficiently}. However, the present work shows a more elegant proof that does not require the geometric mean of matrices nor pseudo-inverses. The simpler expression also gives rise to a more computationally efficient formulation, which also was already noted in \cite{baldwin2023efficiently}. However, the available empirical evidence is only very limited so far. As a second contribution, this work also seeks to close this gap using rigorous empirical testing to compare it with the most efficient previously known methods and comes with code for reproducibility.

The rest of this work is structured as follows. Firstly, the traditional formulation of quantum fidelity is given. Then, preliminaries are shown which are required for the subsequent proof and the actual main theorem is proven. Finally, the efficiency of the more recent method is investigated before the work is concluded.

\section{Quantum fidelity}

Mixed quantum states are mathematically described by density matrices, which are positive-semidefinite (PSD) complex matrices with trace equal to 1. The textbook definition of fidelity between two mixed states $\rho$ and $\sigma$ is
\begin{equation}\label{eq:nielsen_fidelity}
  F(\rho, \sigma) := \Tr\left(\sqrt{\sqrt{\rho}\, \sigma \sqrt{\rho}}\right)^2
\end{equation}
where $\rho$ and $\sigma$ are the density matrices to compare and $\sqrt{\rho}$ denotes the positive square root of $\rho$ \cite{nielsen2010quantum, uhlmann1976transition}. It is also sometimes introduced with the equivalent formulation
\begin{align}\label{eq:uhlmann_fidelity}
  F(\rho, \sigma)
  &= \| \sqrt{\rho} \sqrt{\sigma} \|_1^2
\end{align}
where $\| \cdot \|_1$ is the trace norm \cite{wilde2017qit}.

%

While there are simplification of this formulation for pure states, this work focuses on the most general case, where both density matrices can be mixed states. However, one relevant well-known simplification is that if $\rho$ and $\sigma$ commute, then equation \eqref{eq:nielsen_fidelity} simplifies to
\begin{equation}\label{eq:simplified_fidelity}
  F(\rho, \sigma) = \Tr\left(\sqrt{\rho \sigma}\right)^2
\end{equation}
which can be seen by commuting $\sigma$ with the second $\sqrt{\rho}$ in equation \eqref{eq:nielsen_fidelity}. This work will reconfirm that equation \eqref{eq:simplified_fidelity} holds even if $\rho$ and $\sigma$ do not commute.

\section{Preliminaries}

\subsection{Relation between trace and eigenvalues}

It is well known that the trace of a diagonalizable square matrix is equal to the sum of its eigenvalues. A deeper explanation can be found, for example, in \cite[p. 296]{axler2015linear}.

\begin{lem}\label{lem:Tr_eq_sumlambda} Let $n \in \mathbb N$ and $A \in \mathbb C^{n\times n}$ be diagonalizable. Then
\begin{equation}
  \Tr(A) = \sum_i \lambda_i(A)
\end{equation}
where $\lambda_i(A)$ is the $i$-th eigenvalue of $A$.
\end{lem}

\begin{proof}[Proof \cite{axler2015linear}]
Express the characteristic polynomial $\char(A)$ as
\begin{align*}
    \char(A) &= \det(\lambda I - A)\\
    &= \lambda^n - \Tr(A)\lambda^{n-1} + ... + (-1)^n \det(A)\\
    &= (\lambda - \lambda_1(A))(\lambda - \lambda_2(A))...(\lambda - \lambda_n(A))
\end{align*}
where the second line is the standard form and the last line the factored form of the polynomial. Comparing coefficients of the $\lambda^{n-1}$ term yields the claim.
\end{proof}

\subsection{Cyclic property of the spectrum}\label{sec:cyclic_spectrum}

Quantum fidelity is essentially a sum over eigenvalues. The two $\sqrt\rho$ in \eqref{eq:nielsen_fidelity} can be brought together using a cyclic permutation.

\begin{lem}[Cyclicity of the spectrum]\label{lem:cyclic_spectrum} Let $A, B \in \mathbb C^{n \times n}$ with $n \in \mathbb N$. Then
\begin{equation}
  \sigma(AB) = \sigma(BA)
\end{equation}
where $\sigma(A)$ is the spectrum of $A$.
\end{lem}

\begin{proof}[Proof] It is well-known that the products $AB$ and $BA$ have the same characteristic polynomial \cite{Williamson1954TheCP,schmid1970remark}. Because the spectrum consists of just the roots of the characteristic polynomial, $AB$ and $BA$ must have the same eigenvalues, with the same multiplicities.
\end{proof}

This would already allow to bring the $\sqrt\rho$ in \eqref{eq:nielsen_fidelity} together, if there was not an additional square root. To account for that, it also has to be shown that the cyclic property of the spectrum holds after applying a mapping $f$.

\begin{lem}[Cyclicity of the spectrum with mapping]\label{lem:cyclic_spectrum_mapping} Let $n \in \mathbb N$, $A, B \in \mathbb C^{n \times n}$, their products $AB$ and $BA$ be diagonalizable, and $f$ a continuous function with domain containing $\sigma(AB)$. Then
\begin{equation}
  \sigma(f(AB)) = \sigma(f(BA))
\end{equation}
where $\sigma(A)$ is the spectrum of $A$ and $f$ operating on a matrix is defined via the eigendecomposition.
\end{lem}

\begin{proof} Since $AB$ is diagonalizable and $f$ is defined on $\sigma(AB)$, the transformation $f(AB)$ is applicable. Because $BA$ has the same spectrum as $AB$ (Lemma \ref{lem:cyclic_spectrum}), $f$ is also defined on $\sigma(BA)$, and since $BA$ is also diagonalizable, $f(BA)$ is applicable, as well. Finally, because $f$ is defined to only transform the eigenvalues, the transformed spectra are the same, as well.
\end{proof}


Intuitively speaking, since $AB$ and $BA$ have the same characteristic polynomial and the matrix operation $f$ is defined to be a function only of the roots of that polynomial, the matrices $f(AB)$ and $f(BA)$ have the same (transformed) characteristic polynomial, as well.

Note also that if $A,B$ are PSD matrices then the product $AB$ is diagonalizable with non-negative eigenvalues, as well, despite generally not being PSD itself (except $A$ and $B$ commute) \cite[p. 486]{Horn2012Matrix}. That it has non-negative eigenvalues also follows from Lemma \ref{lem:cyclic_spectrum} by observing that $\sigma((\sqrt{B}\sqrt{A})^\dagger (\sqrt{B}\sqrt{A})) = \sigma(\sqrt{A}B\sqrt{A}) = \sigma(AB)$.



\section{Simplified formulation of quantum fidelity}

\begin{thm} Quantum fidelity, as defined in equation \eqref{eq:nielsen_fidelity}, can be written as
\begin{equation}\label{eq:simplified_fidelity_sqrt}
  F(\rho, \sigma) = \Tr\left(\sqrt{\rho \sigma}\right)^2
\end{equation}
for any two density matrices $\rho$ and $\sigma$.
\end{thm}

\begin{proof} For more concise notation, the positive square root of above will be shown. Firstly, use that the trace is equivalent to the sum of the eigenvalues (Lemma \ref{lem:Tr_eq_sumlambda}).
\begin{align}
  \Tr\left(\sqrt{\sqrt{\rho}\, \sigma \sqrt{\rho}}\right) = \sum_i \lambda_i\left(\sqrt{\sqrt{\rho}\, \sigma \sqrt{\rho}}\right)\label{eq:eig_nielsen_fidelity}
\end{align}
Secondly, since $\sqrt{\rho}\, \sigma \sqrt{\rho} = (\sqrt{\sigma}\sqrt{\rho})^\dagger (\sqrt{\sigma}\sqrt{\rho})$ is PSD and thus diagonalizable with non-negative eigenvalues, $\rho\sigma$ is diagonalizable with non-negative eigenvalues (see section \ref{sec:cyclic_spectrum}), and the square root is continuous on $\mathbb R^+_0$, Lemma \ref{lem:cyclic_spectrum_mapping} can be applied with $A = \sqrt{\rho}\, \sigma$, $B = \sqrt{\rho}$, and $f(x) = \sqrt{x}$.
\begin{equation}\label{eq:after_cyclic}
  \sum_i \lambda_i\left(\sqrt{\sqrt{\rho}\, \sigma \sqrt{\rho}}\right) = \sum_i \lambda_i\left(\sqrt{\rho \sigma}\right)
\end{equation}
Finally, use Lemma \ref{lem:Tr_eq_sumlambda} again.
\begin{equation}
  \sum_i \lambda_i\left(\sqrt{\rho \sigma}\right) = \Tr\left(\sqrt{\rho \sigma}\right)
\end{equation}
Squaring this result again yields the claim.
\end{proof}

\section{Efficiency}

\begin{figure}[h]
  \centering
  \includegraphics[width=0.48\textwidth]{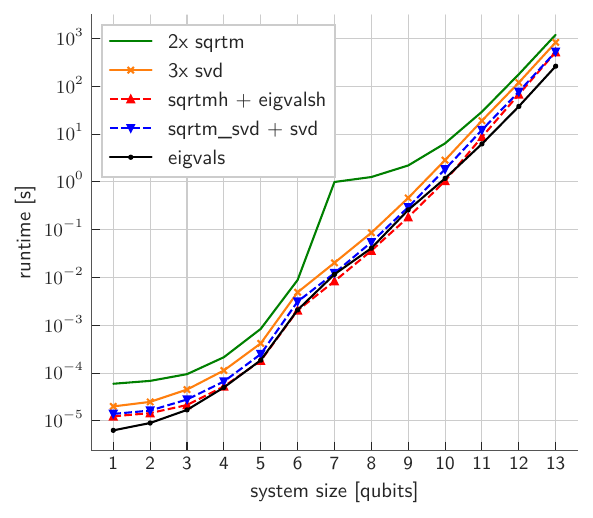}
  \caption{Comparing the performance of different methods to calculate quantum fidelity. ``2x sqrtm'' refers to calculating two general eigendecompositions to calculate the matrix square roots of $\rho$ and $\sqrt{\rho} \sigma \sqrt{\rho}$ in equation \eqref{eq:nielsen_fidelity}. ``3x svd'' is another commonly used method based on the alternative formulation in equation \eqref{eq:uhlmann_fidelity}. ``sqrtmh + eigvalsh'' computes fidelity as $\sum_i \sqrt{\lambda_i(\sqrt{\rho}\sigma\sqrt{\rho})}$, utilizing the specialized eigendecomposition routine for hermitian matrices. ``sqrtm\_svd + svd'' works similarly, but using the SVD. Finally, ``eigvals'' only uses one general routine to calculate eigenvalues, following equation \eqref{eq:efficient_fidelity}. Note that the y-axis is in log-scale. The average values are calculated over $\lceil 10^4 / 2^{k-3}\rceil$ runs, where $k$ is the system size. The Python library Numpy running on an Apple M1 chip have been used for the computation.}
  \label{fig:performance_comparison}
\end{figure}

Applying the spectral mapping theorem \cite[p. 263]{rudin1991fa} to the RHS of equation \eqref{eq:eig_nielsen_fidelity} followed by Lemma \ref{lem:cyclic_spectrum} with $A = \sqrt{\rho}\, \sigma$, $B = \sqrt{\rho}$ provides a computationally more efficient method
\begin{equation}\label{eq:efficient_fidelity}
    F(\rho, \sigma) = \left(\sum_i \sqrt{\lambda_i\left(\rho \sigma\right)}\right)^2
\end{equation}
since it requires only one eigendecomposition. The form \eqref{eq:efficient_fidelity} was already noted before \cite{baldwin2023efficiently}.

To compare the efficiency of the formulation \eqref{eq:efficient_fidelity} with previous methods, it is important to note that there are already more efficient ways to calculate the fidelity based on the usual formulations, compared to applying three general spectral decomposition. Since the eigenvalues of PSD matrices, and thus the matrix square root, can be calculated using the SVD, formulation \eqref{eq:uhlmann_fidelity} can be utilized applying SVD for $\sqrt\sigma$, $\sqrt\rho$, and the trace norm. In formulation \eqref{eq:nielsen_fidelity}, $\sqrt{\rho} \sigma \sqrt{\rho}$ is a hermitian matrix which thus allows to utilize more efficient algorithms to compute eigendecompositions of hermitian matrices provided by common linear algebra libraries. Further, applying the spectral mapping theorem on formulation \eqref{eq:eig_nielsen_fidelity} yields
\begin{equation}
  \sum_i \sqrt{\lambda_i\left(\sqrt{\rho}\, \sigma \sqrt{\rho}\right)}
\end{equation}
which requires eigenvectors only for $\sqrt\rho$. In contrast, equation \eqref{eq:efficient_fidelity} does not need eigenvectors at all, only the eigenvalues of $\rho\sigma$. However, since $\rho\sigma$ is generally not hermitian, the general eigendecomposition has to be used.

For the experiment, $\lceil 10^4 / 2^{k-3}\rceil$ random pairs of density matrices of different sizes corresponding to $k=1,...,13$ qubits have been generated to compare the methods (see Figure \ref{fig:performance_comparison}). The method following equation \eqref{eq:efficient_fidelity} performed significantly better for the smallest and the largest tested random pairs of density matrices (about 2 times faster) and was otherwise almost on par with the best alternative method. This result can probably be further improved by observing that the product $\rho\sigma$ has non-negative eigenvalues (see section \ref{sec:cyclic_spectrum}). All methods discussed so far have a worst-case time complexity of $\mathcal O(n^3)$ in the state dimensionality. However, for more structured or sparse density matrices, the time complexity can be greatly improved, as well \cite{spielman2014nearly}. Furthermore, with quantum computers even an exponential speed-up could be possible \cite{shao2022computing}.

\section{Conclusion}

This work has shown an elegant way to prove that quantum fidelity can be simplified to $F(\rho, \sigma) = \Tr\left(\sqrt{\rho \sigma}\right)^2$ for any two density matrices $\rho$ and $\sigma$. This form is more concise then the usual expression and allows to grasp the symmetry property immediately. Further, a more efficient calculation -- by avoiding any full decomposition -- has been discussed and empirically validated.

Future work might take a look at the consequences of this reformulation on other theorems. Additionally, the computational advantage might allow for faster results where fidelity on mixed states is the computational bottleneck. Finally, advances in reducing the time complexity for calculating eigenvalues can be directly translated into improvements to calculate quantum fidelity.

\section*{Acknowledgements}

I thank Prof. Vikas Garg at Aalto University for support and encouragement. I also thank Jonathan A. Jones, Bartosz Regu\l a, and Danylo Yakymenko for feedback on earlier versions of this work.

\bibliographystyle{quantum}
\bibliography{ref}

\end{document}